\documentclass[copyright,creativecommons]{eptcs} 
\providecommand{\event}{G.Ciobanu, M.Koutny (Eds.): Membrane Computing and\\
Biologically Inspired Process Calculi (MeCBIC 2010)
EPTCS Proceedings, 2010.}

\usepackage{amsmath}
\usepackage{amssymb}
\usepackage{amsthm}
\usepackage{amsfonts}
\usepackage{graphicx}

\usepackage{algebra}

\title{Aspects of multiscale modelling in a process algebra for biological systems}

\author{
\begin{tabular}{ cc }
 Roberto Barbuti  & Giulio Caravagna \\
 Paolo Milazzo & Andrea Maggiolo--Schettini \\
\end{tabular}
\institute{Dipartimento di Informatica,\\ Universit\`a di Pisa,\\Largo Pontecorvo 3, 56127 Pisa, Italy.}
\email{\{barbuti, caravagn, maggiolo, milazzo\}@di.unipi.it}
\and
Simone Tini
\institute{Dipartimento di Informatica e Comunicazione, \\ Universit\`a dell'Insubria, \\ 
Via Mazzini 5, 21100 Varese and Via Carloni 78, 22100 Como, Italy.}
\email{simone.tini@uninsubria.it}
}

\def\titlerunning{Aspects of multiscale modelling in a process algebra for biological systems}
\def\authorrunning{Barbuti et al.}

\date{}
\begin{document}
 \maketitle 

\begin{abstract}
We propose a variant of the CCS process algebra with new features aiming at
allowing multiscale modelling of biological systems. In the usual semantics of
process algebras for modelling biological systems actions are instantaneous.
When different scale levels of biological systems are considered in a single
model, one should take into account that actions at a level may take much more
time than actions at a lower level. Moreover, it might happen that while a
component is involved in one long lasting high level action, it is involved also
in several faster lower level actions. Hence, we propose a process algebra with
operations and with a semantics aimed at dealing with these aspects of
multiscale modelling. We study behavioural equivalences for such an algebra and
give some examples.
\end{abstract}

\section{Introduction}

Formal modelling notations of computer science are nowadays often applied to
the description of biological systems. Such notations can be used to
unambiguously describe the structure and the events governing the dynamics of
the systems of interest, thus allowing development of analysis tools such as
simulators and of formal analysis techniques based, for instance, on model
checking or on behavioural equivalences.

As examples of formalisms that have been applied to the description of
biological systems we mention Bio-PEPA \cite{biopepa1, biopepa2}, the
(stochastic) $\pi$-calculus \cite{priami, spi2, spi3}, Bioambients
\cite{bioambients}, the $\kappa$-calculus \cite{kappa}, the Language for
Biological Systems (LBS) \cite{lbs}, and the Calculus of Looping Sequences (CLS)
\cite{cls1,cls2,cls3}.
In these formalisms the dynamics of a biological system consists of a sequence
of events (usually biochemical reactions) described either as communications
between processes of a process algebra, or as applications of some rewrite
rules. In the stochastic extension of these formalisms, the dynamics of a
system is described by taking also into account the different rates of
occurrence of the events. Rates depend on some parameters associated with the
events (e.g. kinetic constants of the corresponding biochemical reactions) and
on the abundance (or concentration) of the entities (or reactants) that can
cause such events. The rates are then used, as in Gillespie's algorithm
\cite{gillespie}, to describe both the exponentially distributed time elapsing between
two subsequent events, and the probability of an event to occur.

This way of describing the dynamics of biological systems with sequences of
events, however, assumes that the occurrence of one of such events can be described
as an instantaneous change in the system state. In fact, even in the stochastic
approach, the only notion of time that is considered is given by the frequency
of the events, rather than by their duration. The duration of an event is usually 
ignored if it is negligible with respect to the time interval between two events,
or hidden in such a time interval by chosing a rate for the event that is small enough 
to take into account both of its frequency and of its duration.

Phenomena of interest in the study of biological systems often include
processes at different levels of abstraction. A typical example is cell
signalling, that involves gene regulation and protein interaction processes at
the intra-cellular level, a signal diffusion process at the inter-cellular
level, and some macroscopic change at the tissue level. The processes at the
different levels influence each other. However, they involve system
components of very different sizes and are characterized by events having very
different durations. This motivates the {\em multiscale approach to modelling},
whose application to biological systems seems to be promising \cite{multiscale,ABM05,BRS+09,BOA+06,RCS06,SG09}.

The description of the dynamics of biological systems by means of sequences of
instantaneous events is not suitable for multiscale models. When different scale
levels of biological systems are considered in a single model, one should take
into account that events at a level may take much more time than events at a
lower level. Moreover, it might happen that while a component (e.g. a cell) is
involved in a long lasting high level event (e.g. mitosis), it is involved also
in several faster lower level events (e.g. protein sinthesis) that neither
consume such a component nor interrupt the higher level event. Consequently, 
the distinction between the time scales of the events, the possibility of
having the same component involved in several events at different time scales, and the
fact that the completion of some events may (or may not) interfere with other events in which the
same component is involved require new notions of system dynamics to be considered.

In this paper, we propose a process algebra with operations and with a semantics
aimed at dealing with these aspects of multiscale modelling. Actually, we aim
at undertaking a foundational study. Hence, we consider a minimal process algebra
for the description of biological systems (a fragment of CCS \cite{ccs} and of
the Chemical Ground Form \cite{cgf}) and we make the minimal changes we think
to be necessary to describe the new aspects of interest.

As regards the syntax of the new process algebra (called {\em Process Algebra
with Preemptive and Conservative actions}), we propose a new action prefixing
operator that allows an action to be executed in a {\em conservative} (or
non-consuming) manner, namely without removing the process that performed it. As
regards the semantics, we define it as a labelled transition system by following
the {\em ST semantics} approach (see \cite{ST1,ST2,ST3,bravetti}) in which actions are not
instantaneous, but described by two separate starting and ending transitions.
This will permit to have processes in which multiple actions are running in parallel
and competing for their completion. Indeed, we change the usual interpretation of the summation operator (by making it slightly similar to a parallel composition) order to allow a process
to be involved in several actions at the same time. The termination of an action
in a summation may interrupt (in a {\em preemptive} way) the others that are
concurrently executed in the same summation, depending on which action prefixing
operator is used.

The semantics of the process algebra is given in a compositional way, and
this allows us to study behavioural equivalences. In particular, we define
a notion of bisimulation for the process algebra and we prove a congruence
result for it. Some examples are given of use of the process algebra and of the
bisimulation relation.

The paper is structured as follows. In Section \ref{sec:syntax} we introduce the syntax of the
processes, process configurations and some auxiliary functions to define the semantics,
given in Section \ref{sec:semantics} with an example. In Section  \ref{sec:bisimulation} 
a behavioural equivalence for the processes is presented. Finally, we end with some
conclusions and future work in Section  \ref{sec:conclusions}.

\section{A Process Algebra with Preemptive and Conservative actions}\label{sec:syntax}

In this section we present the syntax  of a {\em Process Algebra with Preemptive and Conservative actions}, in the following shortly denoted as PAPC.

\subsection{The Syntax}

PAPC is a process algebra with dyadic communication in the style of CCS \cite{ccs} and of the $\pi$-calculus \cite{pi,pi2}.  Hence, communications involve exactly two processes at a time.  
As in the cases of CCS and of the $\pi$-calculus, the description of biological systems with PAPC is based on a {\em processes-as-molecules} view:
the modeling of a molecular species $S$ occurring $n$ times in a system includes $n$ copies of a process $P_S$ modeling a single molecule.

Differently with respect to the approach of classical process algebras, where actions are instantaneous, in PAPC we assume that actions can consume time and that the instants of start and completion of an action can be detached. An action that has already started but not yet completed is said to be \emph{running}.

Let us assume an infinite set of actions $Act$ ranged over by $\alpha, \beta, \dots$ and a function $(\overline{\hspace{0.1 cm}}): Act \rightarrow Act$ such that $\overline{\overline{\alpha}} = \alpha$. We also denote with $Act_\tau$ the set of actions enriched with the special internal action $\tau$, $Act_\tau=Act \cup \{\tau\}$ .  The abstract syntax of PAPC is a follows.
\begin{definition} Processes of PAPC are defined by the following grammar:
\begin{align*}
P & \;:= \;\;  \nil \agr \alpha.\,P  \agr \alpha:P  \agr P + P \agr P \mid  P \agr A  
\end{align*}
where $\alpha \in \Act$.  We denote the set of all processes by $\cP$.
\end{definition}

As usual, $\nil$ denotes the classical idle process that can perform no action. PAPC processes can perform actions in $\Act$ in two different ways (represented by two different action prefixing operators). Both $\alpha.P$ and $\alpha : P$ can perform the action $\alpha$. However, after performing such an action, the first process simply behaves as $P$, while the second process continues its execution as $P \mid \alpha:P$, namely it produces a copy of $P$, but it is also availabe again to perform $\alpha$. This difference in the behaviour allows $\alpha.P$ to model a molecule that can be involved in a reaction that transforms it into another molecule, and $\alpha:P$ to model a molecule that can be involved in a reaction that does not consume nor transform it.
Another difference between $\alpha.P$ and $\alpha : P$ is related with the fact that actions are not instantaneous and that a process can start several actions concurrently. When an action used as in $\alpha.P$ completes, it interrupts all other running actions for the same process. This agrees with the intuition that $\alpha.P$ represents a transformation of the described molecule into a different one. On the other hand, when an action used as in $\alpha:P$ completes, it does not interrupt the other running actions of the same process. Again, this agrees with the intuition that the described molecule is neither consumed nor transformed into something different. We say that $\alpha$ is {\em preemptive} with respect to $\alpha.P$ and that it is {\em conservative} with respect to $\alpha:P$. Notice that in PAPC the composition of these two approaches is possible, since an action $\alpha$ may be conservative for a process, and the complementary action $\coalpha$ may be preemptive. 

A process $P + Q$ is able to start actions of both $P$ or $Q$, but notice that, because of the difference between conservative and preemptive actions and the fact that actions are not instantaneous, this summation operator does not correspond to the choice operator of classical process algebras. In fact, starting an action in $P$ (resp. $Q$) does not imply that process $Q$ (resp. $P$) is discarded. As a consequence, $P + Q$ may start several actions, which are said to be in \emph{competition}, meaning that they run concurrently until one of them completes. When an action $\alpha$ in $P$ (resp. $Q$) completes, then all actions 
in $Q$ (resp. $P$) competing with it have to be interrupted only if $\alpha$ is preemptive. If this is the case, then $Q$ (resp. $P$) is discarded. Differently, if a conservative action completes, then the summation is not discharged and, consequently, all the running actions in the summation are still running.

A process $P \mid Q$ is the parallel composition of  processes $P$ and $Q$. Handshaking is possible between any action $\alpha$ in $P$ and its complementary action $\coalpha$ in $Q$. The handshaking is to be performed to assign to an instance of communication a unique identifier which may be used to
interrupt the communication. An interesting case here is when an action $\alpha$ in $P$ and its complementary action $\coalpha$ in $Q$ have been coupled with an handshaking, have started and not yet terminated, and some preemptive action in $P$ competing with $\alpha$ completes, so that $\alpha$ must be interrupted.  In fact, in this case also the complementary action $\coalpha$ in $Q$ must be interrupted.

Finally, constants $A$ are used to specify recursive systems. In general, systems are specified as a set of constant defining equations of the form $A \defP P$. As usual, we assume that all processes in these equations are \emph{closed} and \emph{guarded}.

Some further considerations are worth in order to introduce the differences between PAPC and the classical process algebras. The capability of having competing actions is at the basis for the choices we
made in the definition of PAPC. A process $P \defP P_1 + ... + P_n$ can start {\em multiple actions in
parallel}, but can be involved in {\em each action at most once at a time}. Notice that, in classical process algebras, this is not possible since an action, when starts, determines the future process to transform $P$ in. Hence, the choice is resolved at the time of the starting of an action. In this sense, the summation operator of PAPC is not a classical choice for the reason that the competing actions compete for their completion, and, then, the semantics of the completion, and hence the semantics of the PAPC summation, will depend on the type of the action to be completed, namely whether it is conservative or preemptive. 

Consequently, at any time of a computation, $P$ could be in a configuration in which some of its actions are currently running. More precisely, the competition of the running actions is due to the fact that
they are waiting to complete. Practically, the time for completion may be modeled by general distributions
as in \cite{bravetti}, or by delays as in \cite{io,io2}. However, in this first definition of the algebra we do not consider quantitative timing and stochasticity.

In order to define the semantics of PAPC we need to model a process with possibly running actions. We do this by introducing a notion of process configuration.
\begin{definition} Process configurations of PAPC are 
defined by the following grammar:
\begin{align*}
C_P &\;:=\;  \freeze{\alpha}{l}.P  \agr \freeze{\alpha}{l}:P  \agr C_P + C_P \agr C_P \mid  C_P \agr P  
\end{align*}
where $\alpha \in \Act$ and $l \in \mathbb{N}$. We denote the set of all possible process configurations as $\cC$.
\end{definition}

Any process $P\in \cP$ is also in a valid configuration, hence  $\cP \subset \cC$. However, a process configuration may contain actions denoted by a different prefix. In particular, the configuration 
$\freeze{\alpha}{l}.P$ is the configuration reached by $\alpha.\,P$ after $\alpha$ has started, and $\freeze{\alpha}{l}:P$ is the configuration reached by $\alpha :P$ after $\alpha$ has started.
For both the action prefixes, the new argument $l\in \mathbb{N}$ is a natural number that identifies the running action. Notice that these identifiers, which have to be unique, are computed by the handshaking performed before the start of an action and, once a preemptive action is completed, they may be used to interrupt all other competing actions. Then, if one of these competing actions is $\alpha$ and it has started an handshaking with another action $\coalpha$ in a process running in parallel, then also this action $\coalpha$ will be interrupted. This can be obtained by assigning the same identifier to these two actions when the handshaking begins. By the definition of the semantics it will be clear how both the partners will share the same identifier for the actions.

Let us define by structural recursion an auxiliary function $Id: \cC\mapsto \mathbb{N}$ as follows:
\begin{align*}
Id( \freeze{\alpha}{l}.P ) =  Id( \freeze{\alpha}{l}:P ) &= \{l\} \\
Id(C_P + C'_P) = Id(C_P \mid C'_P) &= Id(C_P) \cup Id(C'_P) \\
Id(P) &= \emptyset \,.
\end{align*}
The value $Id(C_P)$ denotes the set of the identifiers of the actions in the configuration $C_P$ that are running. For instance, given a configuration $C_P \defP \freeze{\alpha}{l}.P + \beta:Q | \freeze{\gamma}{l'}.T$, the identifiers collected by function $Id$ are given by  $Id(C_P) = \{l, l'\}$.

Finally, we define a function $Action: \mathbb{N} \times \cC \mapsto \wp(\Act)$ such that $Action(l, C_P)$ collects the set of actions currently running in $C_P$ and with assigned identifier $l$, if any. The function $Action$ is defined as follows:
\begin{align*}
Action( l, \freeze{\alpha}{l}.P ) =  Action( l, \freeze{\alpha}{l}:P ) &= \{\alpha\}  \\
Action( l, \freeze{\alpha}{l'}.P ) =  Action( l, \freeze{\alpha}{l'}:P ) &= \emptyset \quad \text{ if } l \neq l' \\
Action(l, C_P + C'_P) = Action(l, C_P \mid C'_P) &= Action(l,C_P) \cup 
Action(l,C'_P) \\
Action(l, P) &= \emptyset \,.
\end{align*}
The definition of this function is similar to the definition of $Id$.  For  instance, given the configuration
$C_P \defP \freeze{\alpha}{l}.P + \beta:Q | \freeze{\gamma}{l'}.T$, the actions collected by function $Action$ are given by  $Action(l, C_P) = \{\alpha\}$, $Action(l', C_P) = \{\gamma\}$ and, for all $l'' \neq l$ and $l'' \neq l'$,  $Action(l'', C_P) = \emptyset$.

In the next sections we define the semantics of PAPC by using the notions of process, process configurations and these auxiliary functions.

\section{A Structural Operational Semantics for PAPC} \label{sec:semantics}

In this section we define a Structural Operational Semantics (SOS) \cite{plotkin} for PAPC.
The aim of the SOS is to equip PAPC with a Labeled Transition System (LTS), namely a set of transitions of the form $P \stackrel{\ell}{\rightarrow}_r P'$ representing a move from $P \in \cC$ to $P' \in \cC$, with the \emph{label} $\ell$ carrying some information about the move and the index $r$ used to group transitions describing a particular aspect of the behavior of the processes. The LTS is defined by a set of SOS transition rules of the form $\frac{\mbox{\scriptsize{premises}}}{\mbox{\scriptsize{conclusion}}}$. 
Intuitively, each of these rules explains how a move of a process is obtained from moves of its subprocesses. All our rules are in Figures \ref{fig:handshaking}--\ref{fig:rec}. We assume the standard way for assigning an LTS with such a set of transition rules (see, e.g., \cite{aceto}).

The main features of the SOS  we want are the following. Firstly, it must have a mechanism to interrupt competing actions and this mechanism is activated by the completion of a preemptive action. Secondly,  the  style of the semantics must be ST-like, as this permits to easily observe detached events as the start and the completion of an action.

In order to get this features, we define a relation for modeling the start of an action and the coupling of processes; this will be named as the {\em handshaking relation}. Furthermore, we define a {\em completion relation} for modeling the finishing of both preemptive and conservative actions. These two relations will make use of an {\em interruption relation} to model the interruption of currently running actions, as required by the notion of preemptive actions.

\subsection*{The handshaking relation}

This relation is used to model the starting of an action and the coupling of the processes starting complementary actions. The handshaking relation is $\ltrans{}_H  \subseteq  \cC \times \Theta^+ \times \cC$, where $\Theta^+$ contains labels $\theta^+$ of the form
\[
\theta^+ = ( l,  \alpha^+) 
\]
where $l \in \mathbb{N}$ represents the identifier assigned to the started action $\alpha \in Act_\tau$,
and the use of the superscript ``+''  comes from the definition of the semantics in the ST style, in order to denote the start of an action. The SOS rules  in Figure \ref{fig:handshaking} are at the basis of the definition
of $\ltrans{}_H$. We implicitly assume the rules symmetric to $(H_3), (H_4), (H_5), (H_6)$.

\begin{figure}[t]
\centering
\begin{gather*} 
(H_1)\;\; \alpha.P \hand{1}{\alpha} \freeze{\alpha}{1}.P \quad \quad
(H_2)\;\; \alpha:P \hand{1}{\alpha} \freeze{\alpha} {1}:P \\
(H_3)\;\;\frac{P \hand{l}{\alpha} P' \qquad l \not \in Id(Q)}
	{P+Q  \hand{l}{\alpha} P' + Q } \\
(H_4)\;\; \frac{P \hand{l}{\alpha} P' \qquad l \in Id(Q) \quad l'=min\{\mathbb{N} - Id(P+Q) \}}
	{P+Q  \hand{l'}{\alpha} P'[l'/l] + Q }\\
(H_5)\;\; \frac{P  \hand{l}{\alpha} P' \qquad l \not \in Id(Q) \quad \alpha \in Act_\tau}
	{P \mid Q  \hand{l}{\alpha} P' \mid Q }\\
(H_6)\;\; \frac{P  \hand{l}{\alpha}P' \qquad l \in Id(Q) \quad l'=min\{\mathbb{N} - Id(P \mid Q) \} \quad \alpha \in Act_\tau}
	{P \mid Q \hand{l'}{\alpha} P'[l'/l] \mid Q }\\
(H_7)\;\; \frac{P  \hand{l}{\alpha}P' \qquad Q \hand{l'}{\coalpha} Q'  \quad l''=min\{\mathbb{N} - Id(P \mid Q) \}}
	{P\mid Q  \hand{l''}{\tau} P'[l''/l] \mid Q'[l''/l'] }
\end{gather*}
\caption{The handshaking relation $\ltrans{}_H  \subseteq  \cC \times \Theta^+ \times \cC$.}
\label{fig:handshaking}
\end{figure}

Rules $(H_1)$ and $(H_2)$ model the starting of an action $\alpha$. At any time a process with prefix $\alpha$ can start  action $\alpha$ moving to a configuration in which it cannot perform the same action anymore, i.e. the configuration $\freeze{\alpha}{1}.P$ or, analogously, the configuration $\freeze{\alpha}{1}:P$. Such a configuration, together with the one describing the process performing the complementary action, has to be uniquely identified by a natural number representing the identifier of the just started action. 
At this step, the process simply chooses $1$ as unique identifier. All our choices for assigning identifiers to actions are inspired to those of \cite{BG99}, which ensure that the portion of LTS rooted in a given
process is finite. The rules for binary operators $+$ and $\mid$ will solve conflicts of
colliding identifiers, if any. Notice that both preemptive actions and conservative actions start in the 
same way.

Rules $(H_3)$ and $(H_4)$ combine the start of an action with operator $+$. In rule $(H_3)$ the identifier $l$ of the action $\alpha$ started by process $P$ has no conflicts with the identifiers of the competing actions running in process $Q$. Differently, in the case of rule $(H_4)$ a conflict does exist, which implies that a fresh identifier $l'$ replaces $l$. Again, the policy by which we choose the new fresh identifier, along the line of \cite{BG99}, is such that the resulting LTS is finite. More precisely, the use of the set $\mathbb{N} - Id(P+Q)$ is such that we consider, in the process of renaming an identifier in $P'$, the only set of identifiers not used in a process $P+Q$ and, from that, the choice of extracting the minimum value is such that the LTS is finite. We use this strategy in all the rules where we have to resolve some conflicts.

Rules $(H_5)$, $(H_6)$ and $(H_7$) combine the start of an action with the operator $\mid$. Rules $(H_5)$ and $(H_6)$ model an autonomous move by one of the two processes, and deal with identifiers as $(H_3)$ and $(H_4)$, respectively. As in classical process algebras, we do not force $P$ and $Q$ to handshake, since $P$ could handshake with a further process composed in parallel with $P \mid Q$.

Notice that here we may have a conflict even if in $Q$ the action associated with the colliding identifier is the complementary action $\coalpha$.  Rule $(H_7)$ models the handshaking by assigning to this particular instance of synchronization a new fresh identifier $l''$ chosen with the same policy used to resolve conflicts in the previous rules. The renaming of both old identifiers with the newly generated is due 
to the fact that, in general, the two processes will have two different candidate identifiers, i.e. $l$ and $l'$. 
The system in this case exhibits the internal action $\tau^+$. By applying this rule, the two processes terminated this {\em handshaking phase}.

\subsection*{The interruption relation}

This relation, differently from those found in classical process algebras, is used to model the interruption of a set of actions currently running in a process. The interruption is caused by the completion of competing preemptive actions. The interruption relation is $\ltrans{}_I  \subseteq  \cC \times \wp(\mathbb{N}) \times \cC$, where a label $M \in  \wp(\mathbb{N})$ contains the identifiers of the actions that have been interrupted.  
The rules presented in Figure \ref{fig:interruption} are at the basis of the definition of  $\ltrans{}_I$. 

\begin{figure}[t]
\centering
\begin{gather*} 
(I_1)\;\; {\freeze{\alpha}{l}.P \interr{ \{l\}}  \alpha.P } \quad\quad
(I_2)\;\; {\freeze{\alpha}{l}:P \interr{ \{l\}}  \alpha:P } \\
(I_3)\;\;{\freeze{\alpha}{l}.P \interr{ \emptyset } \freeze{\alpha}{l}.P }\quad\quad
(I_4)\;\;{\freeze{\alpha}{l}:P \interr{ \emptyset } \freeze{\alpha}{l}:P }\\
(I_5)\;\;{\alpha.P \interr{ \emptyset }  \alpha.P } \quad\quad
(I_6)\;\;{\alpha:P \interr{  \emptyset }  \alpha:P } \\
(I_7)\;\;\frac{P \interr{L} P' \qquad Q \interr{M} Q' }
{P + Q \interr{ L \cup M} P' + Q' }
\end{gather*}
\caption{The interruption relation $\interr{}  \subseteq  \cC \times \wp(\mathbb{N}) \times \cC$.
}
\label{fig:interruption}
\end{figure}

At any time, a process either in configuration $\freeze{\alpha}{l}.P$ or $\freeze{\alpha}{l}:P$ may interrupt the action it is currently performing. In these cases, treated with rules $(I_1)$ and $(I_2)$, it moves to a
configuration in which the interrupted action $\alpha$ may start again, namely to configuration $\alpha.P$ or $\alpha:P$, respectively. In both the rules, the identifier $l$ of the interrupted action is exhibited
as a label of this transition. Again, this information will be used to interrupt also the partner of this 
action, as we are assuming that there is a partner in the system which, 
after terminating the handshaking phase, has been coupled with the same 
label $l$.

In some cases not all the actions have to be interrupted, so the processes in configuration $\freeze{\alpha}{l}.P$ or $\freeze{\alpha}{l}:P$ must be able also to non-deterministically decide 
whether to interrupt or not. This second case is described by rules $(I_3)$ and $(I_4)$, 
which may seem controversial at first glance. In particular, it may not be clear why a process may independently decide whether to interrupt or not some of the currently running actions. 
The need of this autonomy for the process can be clarified by an example. Let us assume a process configuration  $(P + \Sigma) \mid (Q + \Sigma') \mid S \mid R$, where both $P$ and $Q$ successfully 
complete a preemptive action. The actions to be interrupted are those currently running in both $\Sigma$
and $\Sigma'$, namely those with identifiers denoted by $Id(\Sigma) \cup Id(\Sigma')$. 
Let us assume that some of the actions that have to be interrupted in $\Sigma$ and $\Sigma'$ were coupled with some actions in $S$. In this case, also these actions in $S$ should be interrupted as well.  
Moreover, $S$ may be involved in other actions currently running and coupled with actions in $R$. Indeed, these actions must not be interrupted. This means that from $S$ the correct derivation with the interruption
relation, in general, will not exhibit as label $Id(S)$, indeed it will exhibit a strict subset of $Id(S)$. 
This implies that $S$ must be able to autonomously decide which actions to interrupt, and this can be done by properly combining derivations of the interruption relation. The composition of the relations of the whole semantics will provide the correctness, namely the fact that all and only those to interrupt are 
actually interrupted.

Also, a process which is not performing any action, namely a process in a configuration $\alpha.P$ or $\alpha:P$, does not interrupt any action, as stated by rules $(I_5)$ and $(I_6)$.

Finally, rule $(I_7)$ simply collects the labels of the interrupted actions in a summation. Notice that this relation is not defined for process configurations of the form $P_C \mid P_C'$ as the use of this relation 
is limited to the level of the summation.

\subsection*{The completion relation for preemptive actions}

This relation is used to model the completion of a preemptive action. We will define completion relations also for conservative actions as well as the combination of both preemptive and conservative actions.

The completion relation for preemptive actions is $\ltrans{}_{CP}  \subseteq  \cC \times \Theta_{CP}^-  \times \cC$, with $\Theta^-_{CP}$ containing labels of the form
\[
\theta^- := ( l,  \alpha^-, N) 
\]
where $l \in \mathbb{N}$ represents the identifier that was assigned to the completed action $\alpha \in Act_\tau$ when it was started, $N \in \wp(\mathbb{N})$ is the set of the identifiers of the competing 
actions that are interrupted by the termination of $\alpha$, and the use of the superscript ``-''  comes from the definition of the semantics in the ST style.  
The rules presented in Figure \ref{fig:completionp} are at the basis of the definition of
$\ltrans{}_{CP}$. We implicitly assume rules symmetric to $(C_2)$ and $(C_3)$.

\begin{figure}[t]
\centering
\begin{gather*} 
(C_1)\;\; {\freeze{\alpha}{l}.P \complp{l}{\alpha}{\emptyset} P } \quad\quad
(C_2)\;\; \frac{P \complp{l}{\alpha}{L} P'  }
{P + Q \ \complp{l}{\alpha}{L \cup Id(Q)} P'} \\
(C_3)\;\; \frac{
P \complp{l}{\alpha}{L} P' \quad 
Q \interr{M} Q' \quad  M \supseteq (Id(Q) \cap L) 
\quad V = (L \cup M) \setminus (L \cap Id(Q))  \quad \alpha \in Act _\tau
}
{P \mid Q \complp{l}{\alpha}{V} P' \mid Q'}\\
(C_4)\;\; 
\frac{
P \complp{l}{\alpha}{L} P' \quad Q \complp{l}{\coalpha}{M} Q' \quad
N = L \cap M \quad
}
{P \mid Q \complp{l}{\tau}{(L \cup M) \backslash N} P' \mid Q'}
\end{gather*}

\caption{The completion relation for preemptive actions $\ltrans{}_{CP}  \subseteq  \cC \times  \Theta_{CP}^-  \times \cC$.}
\label{fig:completionp}
\end{figure}

Rule $(C_1)$ describes the completion of a preemptive action. When it completes, as the action is preemptive, the process is substituted by its continuation $P$. In the label, the 
identifier $l$ is needed to couple this process with the one performing the corresponding complementary action $\overline{\alpha}$, which will have the same identifier $l$ because of the handshaking, and $\emptyset$ states that no action is interrupted.

Rule $(C_2)$ states that the completion of a preemptive action in $P$  affects a summation $P+Q$ so that all actions running in $Q$ should be interrupted. This is obtained by adding to the set of labels of actions interrupted $L$, the set of actions currently running in the process $Q$ which disappears by the completion of the action in $P$, hence the exhibited set of labels becomes $L \cup Id(Q)$. 

Rule $(C_3)$ states that the completion of a preemptive action in $P$ affects a parallel composition $P \mid Q$ so that all actions running in $Q$ that are coupled with actions interrupted in $P$, must be interrupted as well.

Rule $(C_4)$ models the case in which both $P$ and $Q$ complete preemptive actions that were coupled.
As in classical process algebras, the whole system $P \mid Q$ exhibits an internal  action $\tau$. 
Some of the actions required to be interrupted outside $P$, may be also required to be interrupted by $Q$. 
Such a set is denoted by $N$ and can be removed from the set of actions that can be interrupted outside $P \mid Q$. The remaining set of actions, which have to be still interrupted by further composition with the parallel operator outside $P\mid Q$, is the set of those belonging to $P$ and not to $Q$, and viceversa.

\subsection*{The completion relation for conservative actions}

This relation is used to model the completion of a conservative action. 
This  relation is $\ltrans{}_{CC}  \subseteq  \cC \times \Theta_{CC}^-  \times \cC$, where $\Theta_{CC}^-$ contains labels of the form
\[
\theta^- := ( l,  \alpha^-, N ,P) 
\]
where $l \in \mathbb{N}$ represents the identifier assigned to the completed action $\alpha \in Act_\tau$, $N \in \wp(\mathbb{N})$ is the set of identifiers of the interrupted actions, and $P \in \cP$ is the continuation of the action which terminated and that, syntactically, must be propagated at the level of a parallel composition. At first sight it could sound strange that we need the component $N$. The idea is that we have to take care that a process $Q$ terminating a conservative action $\alpha$ could be composed in parallel with another  process $Q'$ terminating the action $\coalpha$ coupled with $\alpha$. Now, if $\coalpha$ is preemptive, there may be some running actions $\beta$ in $Q'$ that should be interrupted, which implies that if there is an action $\overline{\beta}$ in $Q$ coupled with $\beta$, also $\overline{\beta}$ must be interrupted. For this reason, such a $\overline{\beta}$ must appear in $N$. Of course, if also $\coalpha$ is conservative, then in the transition by $Q$ used to infer the transition of $Q \mid Q'$,  $N$ will be empty.

The rules presented in Figure \ref{fig:completionc} are at the basis of the definition of
relation $\ltrans{}_{CP}$. We implicitly assume rules symmetric to $(C_6)$ and $(C_7)$.

\begin{figure}[t]
\centering
\begin{gather*} 
(C_5)\;\; {\freeze{\alpha}{l}:P \complc{l}{\alpha}{\emptyset}{P} \alpha:P } \quad \quad
(C_6)\;\;\frac{P \complc{l}{\alpha}{L}{P''} P'  \quad Q \interr{M} Q'}
	{P+Q \complc{l}{\alpha}{L \cup M}{P''} P' + Q' }\\
(C_7)\;\; \frac{
P \complc{l}{\alpha}{L}{P''} P' }
{P \mid Q \complc{l}{\alpha}{L}{P''} P' \mid Q}
\quad
(C_8)\;\;\frac{P \complc{l}{\alpha}{\emptyset}{P''} P' \qquad Q \complc{l}{\coalpha}{\emptyset}{Q''} Q' }
	{P\mid Q \complp{l}{\tau}{\emptyset} P' \mid Q' \mid P'' \mid Q'' }\quad
 \end{gather*}

\caption{The completion relation for conservative actions $\ltrans{}_{CC}  \subseteq  \cC \times  \Theta_{CC}^-  \times \cC$.}
\label{fig:completionc}
\end{figure}

Rule $(C_5)$ deals with termination of a conservative action $\alpha$ in configuration $\freeze{\alpha}{l}:P$. The process becomes able to perform the terminated action again, namely it rolls back to configuration ${\alpha}:P$. Also, as expected by a conservative action, it produces the continuation $P$, which, because of the inductive approach of the SOS semantics, cannot appear at the same syntactic level of the configuration ${\alpha}:P$. Specifically, the continuation $P$ will have to appear at the level of a 
parallel composition. To forward $P$ at the correct syntactic level, $P$ is exhibited
as a label of the transition. The empty set used in the label denotes that no action is interrupted by
completion of $\alpha$.

Rule $(C_6)$ clearly justifies the terminology ``conservative''.  When a conservative action $\alpha$ completes in a process $P$ being part of a configuration $P+Q$, it is required neither that actions in $Q$ are interrupted, nor that $Q$ is canceled. This is clearly different from what happens when a preemptive action is completed (rule $(C_2)$). More precisely, the set of actions interrupted in $Q$, namely $M$, is a subset of $Id(Q)$ since here not all the actions in $Q$ have to be interrupted. Furthermore, the continuation of the action, namely $P''$ in the rule, is exhibited as a transition label, since also in this case we are not yet at the syntactic level of a parallel composition. 

Rule $(C_7)$ describes the case in which $P$ completes a conservative action $\alpha$ and the coupled action $\coalpha$ is not in $Q$, since it runs in some other process running in parallel with $P \mid Q$.

Rule $(C_8)$ deals with the completion of two coupled conservative actions. As expected,  the system exhibits an internal action $\tau$, the label shows that no action has to be interrupted, and 
both the continuations appearing in the labels of the transitions of both processes $P$ and $Q$, namely $P''$ and $Q''$, are put in parallel with the continuations $P'$ and $Q'$.  Notice that this last derivation is a derivation for $\rightarrow_{CP}$ rather than $\rightarrow_{CC}$. The reason for this is that $P''$ and $Q''$ are already at the correct syntactic level and do not require to be lifted anymore.

\subsection*{Completion of both conservative and preemptive actions}

We have to deal with the completion of two coupled actions $\alpha$ and $\coalpha$ such that one of them is conservative and the other preemptive. To this purpose, we add the rule shown in Figure \ref{fig:completionh} and we implicitly assume a symmetric rule.

\begin{figure}[t]
\centering
\begin{gather*} 
(C_9) \;\; \frac{P \complc{l}{\alpha}{L}{P''} P' \qquad Q \complp{l}{\coalpha}{M} Q'  \quad L \subseteq M}
	{P\mid Q \complp{l}{\tau}{M \backslash L} P' \mid Q' \mid P''} 
\end{gather*}
\caption{The completion relation for hybrid actions obtained by means of the other completion relations.}
\label{fig:completionh}
\end{figure}

Notice that $L \subseteq M$ expresses that the actions running in the process $Q$ performing
the conservative actions and that have to be interrupted are those that were coupled with actions running in the process $P$ performing the preemptive action. In fact, such a coupling is the only reason we have to interrupt actions running in $Q$.

\subsection*{Recursive definitions}

As far as the naming of processes is concerned, we define the standard rule for recursion, showed in Figure \ref{fig:rec}, for all the relations we defined. 
\begin{figure}[t]
\centering
\begin{gather*} 
(R_1) \;\; \frac{P \stackrel{\ell}{\rightarrow_r} P'}
	{A  \stackrel{\ell}{\rightarrow_r} P'} \qquad \mbox{if } A \defP P
\end{gather*}
\caption{The standard rule for recursion.}
\label{fig:rec}
\end{figure}

\subsection*{A toy example}

In this section we discuss the modeling of a multiscale system where we consider two populations. At a higher level of abstraction we consider a cell $C$, and at a lower level a generic protein $P$. A cell $C$  can be involved in a process leading to its duplication. Also, it can be involved in some low--level reactions (i.e. DNA transcription inside its nucleus) leading to the creation, in the environment outside $C$, of a protein of species $P$.  Of course, we consider this model at a level of detail such that we do not need to take into account any other possible population of either cells or proteins which could be involved in the dynamics.

In the context of chemically reacting systems such a system may be described by two populations $C$ and $P$, and by two reactions $R_1$ and $R_2$ such that
\begin{align*}
R_1: \; C & \ltrans{} C + C  & R_2:\; C \ltrans{} C + P.
\end{align*}
Reactions $R_1$ and $R_2$ model the non linear growth of cell $C$ and the production of a protein $P$ by a cell $C$, respectively. Notice that, at this level of detail,  the production of protein $P$ depends on the cell $C$ where all the details of the biological process leading to the creation of the protein are abstracted away. The initial state of the system can be defined to have a precise initial number of cells $C$ and proteins $P$ in the environment. 

We model now such a system in PAPC, and we show how the semantics models the behavior of the populations. As PAPC is based on the paradigm processes-as-molecules, we start by assuming two types of processes for each species which, for clarity, are named $C$ and $P$. As we want to model two reactions, we assume the following set of actions $\{ \alpha, \coalpha, \gamma ,\cogamma \}$ where $\alpha$ (resp. $\coalpha$)  and $\gamma$ (resp. $\cogamma$) model reaction $R_1$ and $R_2$, respectively. Also, as in PAPC the communication is dyadic and the reactions use a single reactant, we define two auxiliary process $\Aux$ and $\Bux$, used to model the communications on $\alpha$ and $\gamma$, respectively.

Reaction $R_1$ creates two new different cells able to start again, if possible, the duplication process. In the context of PAPC we model $R_1$ by using preemptive actions for both $\alpha$ and $\coalpha$ since the duplication of a cell interrupts all the low-level protein-transcription event inside the duplicated cell. Differently, the actions modeling reaction $R_2$ are conservative since the protein-transcription event does not interrupt the duplication process started by a cell.

The PAPC processes are defined as 
\begin{align*}
C & \defP \alpha.(C \mid C) + \gamma : P \qquad \quad
\Aux \defP \coalpha.(\Aux \mid \Aux) \qquad \quad
\Bux  \defP \cogamma : \nil \,.
\end{align*}
Notice that we do not give a definition of the process $P$ since it appears only as a product of the events we want to model, and hence we are not interested in the interactions it may have in the system.

Process $C$ as expected can perform two actions. Action $\alpha$ produces the two new copies of $C$. Such an action is performed by synchronizing with the auxiliary process $\Aux$ which will produce two copies of itself to permit the duplication of the new cell. Indeed, the number of copies of the auxiliary process $\Aux$ must grow with the same law of growth for the cells $C$. Both $C$ and $\Aux$ behave as preemptive in $\alpha$ and $\coalpha$, as expected. Process $C$ can also perform, by synchronizing with $\Bux$, the action $\gamma$. The result of such action is to produce a new protein $B$ without interrupting its running duplication, if any. The number of processes $\Bux$ in the system bounds the number of cells which can simultaneously produce a protein of type $B$. This is obtained by producing, in $\Bux$, the nil process. If this were a too strong constraint, it would have been possible to add a $\Bux$ process in the continuation of the action $\alpha$ in $C$, as done for $\Aux$. For the sake of simplicity in this model we consider this constraint to be reasonable.

We discuss now some features of the semantics of PAPC for a simple system $S$ described by the process
\[
S \defP C \mid \Aux \mid \Bux\, .
\]
In $S$ both the reactions may fire. We assume reaction $R_1$ to fire first. To this extent, the semantics permits to observe handshaking derivations as 
\[
C \hand{1}{\alpha} \freeze{\alpha}{1}. ( C\mid C) + \gamma:P \qquad \quad \Aux \hand{1}{\alpha} \freeze{\coalpha}{1}.(\Aux \mid \Aux) \,.
\]
Then the whole process $S$ performs a derivation as
\[
S \hand{1}{\tau} \freeze{\alpha}{1}. ( C\mid C) + \gamma:P \mid \freeze{\coalpha}{1}.(\Aux \mid \Aux) \mid \Bux \equiv S'
\]
where the new process $S'$ is such that the action $\alpha$ is now running in $C$ and in $\Aux$, with identifier $1$. In $S'$ action $\alpha$ can not start, but just complete, we assume reaction $R_2$ to fire. The semantics permits to observe the handshaking derivations 
\[
\gamma:P \hand{1}{\gamma} \freeze{\gamma}{1}:P \qquad \quad \Bux \hand{1}{\cogamma} \freeze{\cogamma}{1}:\nil \,.
\]
The composition of these derivations resolves the conflicts of the colliding identifiers such that the derivation for $C$ will be 
\[
\freeze{\alpha}{1}. ( C\mid C) + \gamma:P \hand{2}{\gamma}  \freeze{\alpha}{1}. ( C\mid C) + \freeze{\gamma}{2}:P
\] 
and the whole system performs the following derivation
\[
S' \hand{2}{\tau} \freeze{\alpha}{1}. ( C\mid C) +  \freeze{\gamma}{2}:P \mid \freeze{\coalpha}{1}.(\Aux \mid \Aux) \mid  \freeze{\cogamma}{2}:\nil \equiv  S''
\]
where in $S''$ all the  possible actions are running. We consider now two different cases:  $(a)$ $R_1$ completes before $R_2$ and $(b)$ viceversa.
\begin{itemize}
\item[$(a)$] Reaction $R_1$ completes before $R_2$: in this case action $\alpha$ (resp. $\coalpha$) completes before action $\gamma$ (resp. $\cogamma$), interrupting it. The semantics permits to derive transitions as 
\[
\freeze{\alpha}{1}. ( C\mid C) \complp{1}{\alpha}{\emptyset} C \mid C \qquad \quad \freeze{\coalpha}{1}.(\Aux \mid \Aux) \complp{1}{\coalpha}{\emptyset}\Aux \mid \Aux \,.
\]
Actions in $C$ and in $\Bux$ have to be interrupted, hence we derive 
\[
\freeze{\alpha}{1}. ( C\mid C) +  \freeze{\gamma}{2}:P \complp{1}{\alpha}{\{2\}} C \mid C \qquad \quad
\freeze{\cogamma}{2}:\nil \interr{\{2\}} \Bux
\]
where $\{2\}$ denotes the actions to be interrupted. Consequently, the whole process $S''$ will perform the transition
\[
S'' \complp{1}{\tau}{\emptyset} C \mid C \mid \Aux \mid \Aux \mid \Bux
\]
where in the resulting  process no actions are running, as expected, and there are two cells and two auxiliary processes $\Aux$.

\item[$(b)$] Reaction $R_2$ completes before $R_1$: in this case action $\gamma$ (resp. $\cogamma$) completes before action $\alpha$ (resp. $\coalpha$). The semantics permits to derive transitions as
\[
\freeze{\cogamma}{2}:\nil \complc{2}{\cogamma}{\emptyset}{\nil} \cogamma:\nil \qquad \quad 
\freeze{\gamma}{2}:P \complc{2}{\gamma}{\emptyset}{P} \gamma:P\, .
\]
No actions have to be interrupted in any process, hence we derive 
\[
\freeze{\alpha}{1}. ( C\mid C) \interr{\emptyset} \freeze{\alpha}{1}. ( C\mid C) \qquad \quad
\freeze{\alpha}{1}. ( C\mid C) + \freeze{\gamma}{2}:P \complc{2}{\gamma}{\emptyset}{\nil} \freeze{\alpha}{1}. ( C\mid C) + \gamma:P \,. 
\] 
The whole process $S''$ will then perform the transition
\[
S'' \complp{2}{\tau}{\emptyset}  \freeze{\alpha}{1}. ( C\mid C) + \gamma:P \mid \freeze{\coalpha}{1}.(\Aux \mid \Aux) \mid \Bux \mid P \mid \nil 
\]
where, as expected, in the resulting process only one action is still running (cell division) and a single protein $P$ has been produced.
\end{itemize}

\section{Bisimulation equivalence for PAPC} \label{sec:bisimulation}

Bisimulation equivalence is a central notion in concurrency theory. For processes with a higher order behaviour, namely processes whose behaviour is described by portions of transition systems in which processes appear in the labels, the notion of bisimulation is usually replaced by a higher order notion of bisimulation \cite{hob1,hob2,hob3}, which can be rephrased in our setting as follows:
\begin{definition}
A symmetric relation ${\mathcal R} \subseteq \cC \times \cC$ is a \emph{bisimulation} iff whenever $(P,Q) \in {\mathcal R}$, then it holds that:
\begin{itemize} 
\item if $P \stackrel{\ell}{\rightarrow_{r}} P'$ for any $P' \in \cC$, $r \in \{H,I,CP\}$ and label $\ell$, then
$Q \stackrel{\ell}{\rightarrow_{r}} Q'$ for some $Q' \in \cC$ such that $(P',Q') \in {\mathcal R}$.
\item  if $P \complc{l}{\alpha}{L}{P''} P'$ for any $P' \in \cC$, $l \in \mathbb{N}$, $\alpha \in Act$, $L \subseteq \mathbb{N}$ and  $P'' \in \cC$, then $Q \complc{l}{\alpha}{L}{Q''} Q'$
for some $Q',Q'' \in \cC$ such that $(P',Q') \in {\mathcal R}$ and $(P'',Q'') \in {\mathcal R}$.
\end{itemize}
\end{definition}
The union of all bisimulations is, in turn, a bisimulation, which is denoted with $\approx$ and is called as ``the bisimulation''.

For an algebric treatment of bisimulation equivalence and to reason in a compositional way, a bisimulation is required to be a congruence. By taking the standard notion of context $C[\;]$, a bisimulation ${\mathcal R}$ is a congruence with respect to all operations of the process algebra if and only if, given any pair $(P,Q) \in {\mathcal R}$ and any context $C[\;]$, it holds that $(C[P],C[R]) \in {\mathcal R}$.

\begin{theorem}
Bisimulation is a congruence w.r.t. all PAPC operations.
\end{theorem}
\begin{proof}
If we consider PAPC without recursion, then the proof comes for free. In fact, in \cite{mgr05} it is shown that higher order bisimulation is a congruence with respect to all process algebra operations whose semantics
is defined through transition rules respecting some syntactical constraints, and it can be checked that the transition rules we use respect such constraints. The extension of the proof to the case of recursion is standard.
\end{proof}

The bisimulation relation for PAPC is a very fine behavioral equivalence. 
This can be seen as a disadvantage, since with behavioral equivalences it is often desirable to be able to equate as many processes as possible.
On the other hand, the fact that bisimulation turns out to be fine may have another meaning, namely that all the ingredients used in the process algebra play an important role.
This does not happen, for instance, for the parallel composition in some variants of CCS where it can be reduced into an equivalent summation of processes.

We go through this last consideration via an example. Let us consider the following two PAPC processes
\begin{align*}
C &\defP \alpha. (C \mid C) & C' \defP \alpha:C'\, .
\end{align*}
The behaviour of the two processes is similar: both of them perform an action $\alpha$ and then continue as with copies of the initial process (in the case of $C'$ one of such two copies is not indicated in the continuation of $\alpha$ since it is obtained from the semantics of conservative actions).
Even if the behaviour of the two processes seems to be the same, it is immediate to see that the two processes are not bisimilar. In fact, they repeatedly perform the same handshaking transition $\ltrans{l,\alpha^+}_H$, but followed by two different completion transitions, namely $\ltrans{l,\alpha^-,\emptyset}_{CP}$ and $\ltrans{l,\alpha^-,\emptyset,C'}_{CC}$.
This permits to state that, in general, $C \not \approx C'$. This is exactly what we expect from our bisimulation relation since, when the two processes are
put in a summation context, their behavior would determine the behavior of the whole context. In fact,
the completion of the action in a process $C+\Sigma$ would interrupt any action currently running in $\Sigma$ and, differently, for the case of $C'+\Sigma$ no actions in $\Sigma$ would be interrupted. This is
due to the fact that the two processes perform the same action but with different prefix operators.

In order to see whether this is important, is enough to consider the toy example given in the previous section where actions are modeled by a process $C \defP \alpha.(C \mid C) + \gamma : P$. In that example the creation of two new cells should be an event which interrupts, when completes, the production of protein $P$ by the cell. Of course, if the process used in the model would have been $C \defP \alpha: C + \gamma : P$, then the completion of the duplication process for the cell would not have interrupted the production of $P$.

\section{Conclusions} \label{sec:conclusions}

In this paper we considered the problem of modeling biological systems in which
different scale levels are taken into account resulting in the fact that  actions at a 
higher level may take more time than actions at a lower one. 

In order to model such systems we defined PAPC, a variant of the CCS process
algebra in which a process may be simultaneously involved in one long lasting
high level action and in several faster lower level actions. In order to model this,
we added in the algebra two different prefix operators to model the role of a process
in an action. A process can either act as conservative or preemptive in an action,
resulting in two different behaviors for the process and the other actions which
currently are started and not completed. 

We gave a compositional Structural Operations Semantics for PAPC by means of different
relations, one for each of the possible events which may change the state of a 
process, namely the start and the completion of the actions. The semantics
we gave is in ST style as this permits to observe the
start and the completion of an action as two detached events. This style of the
semantics permits also to observe processes in configurations in which multiple actions
are started and not completed.

We also defined a notion of behavioral equivalence for PAPC processes
based on the ideas of higher-order bisimulations for process calculi. We proved
that our bisimulation is a congruence for all PAPC operators.

In the paper, we also showed some simple example of PAPC processes such that
their semantics permits to observe the key features of the algebra. We also discussed
the notion of bisimulation we introduced by analyzing two simple PAPC processes.

As a future work we will apply PAPC to the modeling of multiscale systems in order
to prove the utility of the formalism. Also, we may consider to enrich PAPC with
biologically inspired operators to easily model complexation, de-complexation or
more complex biological structures as membranes or compartments as it has
been previously done with other calculi. Moreover, we may consider
the definition of more biologically inspired notions of equivalence for PAPC processes.

\bibliographystyle{eptcs} 

\begin{thebibliography}{1}

\bibitem{aceto} L. Aceto, W.J. Fokkink and C. Verhoef (2001):
\newblock \emph {Structural operational semantics}. 
\newblock{ \sl{Chapter in: J.A. Bergstra, A. Ponse and S.A. Smolka (Eds.): Handbook of Process Algebra}, Elsevier, pp. 197--292.} 

\bibitem{ABM05} T. Alarcon, H. M. Byrne and P. K. Maini (2005): 
\newblock \emph {A multiple scale model for tumour growth}. 
\newblock{ \sl{MultiscaleModel Sim.} 3, 440--475}.

\bibitem{hob1} E. Astesiano, A. Giovini and G. Reggio (1988):
\newblock\emph{Generalized bisimulation in relational specifications}.
\newblock{ \sl{Proc. STACS 98, Springer LNCS} 294, 207--226}.

\bibitem{multiscale} G.S. Ayton, W.G. Noid and G.A. Voth (2007):
\newblock\emph{Multiscale modeling of biomolecular systems: in serial and in parallel}.
\newblock{ \sl{Current Opinion in Structural Biology,} 17, Issue 2, 192-198}.

\bibitem{io} R. Barbuti, G. Caravagna, A. Maggiolo-Schettini and P. Milazzo (2009):
\newblock \emph {On the Interpretation of Delays in Delay Stochastic Simulation of Biological Systems}.  2nd Int. Workshop on Computational Models for Cell Processes (CompMod'09),  \newblock{ \sl{EPTCS} 6, 17--29}. 

\bibitem{io2} R. Barbuti, G. Caravagna, A. Maggiolo-Schettini and P. Milazzo (2010):
\newblock \emph {Delay Stochastic Simulation of Biological Systems: A Purely Delayed Approach}.  \newblock{Submitted}. 

\bibitem{cls2} R. Barbuti, G. Caravagna, A. Maggiolo-Schettini, P. Milazzo and G. Pardini (2008):
\newblock \emph {The Calculus of Looping Sequences}.
\newblock{ \sl{Chapter in: M.Bernardo, P.Degano and G.Zavattaro (Eds.): Formal Methods for Computational Systems Biology (SFM 2008), Springer LNCS} 5016, 387--423}.

\bibitem{cls3} R. Barbuti, A. Maggiolo-Schettini, P. Milazzo and G.Pardini (2008):
\newblock \emph {Spatial Calculus of Looping Sequences}.
\newblock{ \sl{Int. Workshop From Biology to Concurrency and Back (FBTC'08), ENTCS} 229(1), 21--39}. 

\bibitem{BRS+09} F. Billy, B. Ribba, O. Saut , H. Morre-Trouilhet, T. Colin, D. Bresch, J.P. Boissel, E.Grenier and J.P. Flandrois. (2009):
\newblock \emph {A pharmacologically based multiscale mathematical model of angiogenesis and its use in investigating the efficacy of a new cancer treatment strategy}
\newblock{\sl{J.Theor.Biol.} 260, 545-562}.

\bibitem{hob2}G. Boudol (1989):
\newblock \emph{Towards a lambda-calculus for concurrent and communicating systems}.
\newblock{ \sl{Proc. TAPSOFT, Springer LNCS} 351, 149--161}. 

\bibitem{BG99} M. Bravetti and R. Gorrieri (1999): 
\newblock \emph{Deciding and Axiomatizing ST Bisimulations for a Process Algebra
with Recursion and Action Refinement}.
\newblock{ \sl{Tech, Rep. UBLCS-99-1, University of Bologna}}.

\bibitem{bravetti} M. Bravetti and R. Gorrieri (2002): 
\newblock \emph {The theory of interactive generalized semi-Markov processes}.   
\newblock{ \sl{Theoretical Computer Science} 282 (1), 5--32}.


\bibitem{BOA+06}  H.M. Byrne, M.R. Owen., T. Alarcon, J. Murphy and P.K. Maini (2006): 
\newblock \emph{Modelling the response of vascular tumours to chemotherapy: a multiscale approach}. 
\newblock{ \sl{Math.Mod.Meth. Appli. Sci.} 15, 1219--1241}.


\bibitem{cgf} L. Cardelli (2008):
\newblock \emph {On Process Rate Semantics}.
\newblock{ \sl{Theoretical Computer Science} 391(3) 190--215}.

\bibitem{biopepa1} F. Ciocchetta and J. Hillston (2009):
\newblock \emph{Bio-PEPA: a Framework for the Modelling and Analysis of Biochemical Networks}. \newblock{ \sl{Theoretical Computer Science} 410 (33-34), 3065--3084}. 

\bibitem{biopepa2} F. Ciocchetta and J. Hillston (2008):
\newblock \emph{Calculi for Biological Systems}.
\newblock{ \sl{Chapter in: M.Bernardo, P.Degano and G.Zavattaro (Eds.): Formal Methods for Computational Systems Biology (SFM 2008)}, Springer LNCS 5016, 265--312}.

\bibitem{kappa}  V. Danos, J. Feret, W. Fontana, R. Harmer and J. Krivine (2007): 
\newblock \emph{Rule-Based Modelling of Cellular Signalling}.
\newblock { \sl {Proceedings of CONCUR07, Springer LNCS}, 17--41}.

\bibitem{gillespie} D. Gillespie (1977): 
\newblock \emph{Exact Stochastic Simulation of Coupled Chemical Reactions}. 
\newblock { \sl { Journal of Physical Chemistry} 81, 2340}.

\bibitem{ST2} M. Hennessy (1988):
\newblock \emph{Axiomatising Finite Concurrent Processes}.
\newblock{ \sl{ SIAM Journal of Computing} 17(5), 997--1017}. 

\bibitem{ST3} R.J. van Glabbeek (1990):
\emph{The Refinement Theorem for ST Bisimulation Semantics}.
\newblock{ \sl{Proc. IFIP Working Conference on Programming Concepts and Methods}.
North Holland}.

\bibitem{ST1} R.J. van Glabbeek and F.W. Vaandrager (1987):
\newblock \emph{Petri Net Models for Algebraic Theories of Concurrency}.
\newblock{ \sl{Proc. PARLE, Springer LNCS} 259, 224--242}.


\bibitem{cls1} P. Milazzo (2007):
\newblock \emph{Qualitative and Quantitative Formal Modeling of Biological Systems}.
\newblock{ \sl{Ph.D. Thesis, Department of Computer Science, University of Pisa}}.

\bibitem{ccs} R. Milner (1989): \newblock 
\emph{Communication and Concurrency}. 
\newblock{ \sl{Prentice Hall, International Series in Computer Science}, ISBN 0-131-15007-3, 1989}.

\bibitem{pi}R. Milner, J. Parrow and D. Walker (1992): \newblock
\emph{A Calculus of Mobile Processes, I} 
\newblock{ \sl{Inf. Comput.} 100(1): 1--40}.

\bibitem{pi2}R. Milner, J. Parrow and D. Walker (1992): \newblock
\emph{A Calculus of Mobile Processes, II} 
\newblock{ \sl{Inf. Comput. 100(1)}: 41--77}.

\bibitem{mgr05}
M. Mousavi, M. Gabbay and M.A. Reniers (2005): 
\newblock \emph{SOS for Higher Order Processes}. 
\newblock{ \sl{Proc. CONCUR 2005, Springer LNCS} 3653, 308--322}.

\bibitem{lbs} M. Pedersen and G. Plotkin (2010):
\newblock \emph{A Language for Biochemical Systems: Design and Formal Specification}. 
\newblock{ \sl{Transactions on Computational Systems Biology XII}, 5945(3),  77--145}.

\bibitem{plotkin} G.D. Plotkin (1981): \newblock \emph{A Structural Approach to Operational Semantics}. 
\newblock{ \sl{Tech. Rep. DAIMI} FN-19, Aarhus University, Denmark}. 

\bibitem{priami} C. Priami (1995): \newblock \emph{Stochastic $\pi$-Calculus}. 
\newblock{ \sl{The Computer Journal} 38 (7), 578--589}.

\bibitem{spi2} C. Priami, A. Regev, E. Shapiro and W. Silverman (2001): 
\newblock \emph{Application of a stochastic name-passing calculus to 
representation and simulation of molecular processes}. 
\newblock{ \sl{Information Processing Letters}, 80:25--31}. 

\bibitem{bioambients} A. Regev, E. M. Panina, W. Silverman, L. Cardelli
and E. Y. Shapiro (2004):
\newblock \emph{Bioambients: an abstraction for biological compartments}.
\newblock{ \sl{Theoretical Computer Science}, 325(1):141--167}.

\bibitem{spi3} A. Regev, W. Silverman and E. Shapiro (2001): 
\newblock \emph{Representation and simulation of biochemical  processes using the 
pi-calculus process algebra}. 
\newblock{ \sl{In Pacific Symposium on Biocomputing} 6,  459--470. World Scienti?c Press}. 

\bibitem{RCS06} B. Ribba, T. Colin and S. Schnell (2006): 
\newblock \emph{A multiscale mathematical model of cancer, and its use in analyzing irradiation therapies}.
\newblock {\sl {Theor.Biol.Med.Model}.3, 7}.

\bibitem{SG09} Sloot, P. M. A. and A. G. Hoekstra (2009): 
\newblock \emph{Multi-scale modelling in computational biomedicine}
\newblock {\sl {Brief. Bioinform.}, doi:10.1093/bib/bbp038}.

\bibitem{hob3}
B. Thomsen (1995):
\newblock
\emph{A theory of Higher Order Communicating Systems}.
\newblock{ \sl{Inform. Comput.} 116, 38--57}.

\end{thebibliography}

\end{document}